\documentclass[a4paper,twocolumn,11pt,accepted=2020-05-20]{quantumarticle}
\pdfoutput=1
\usepackage[utf8]{inputenc}
\usepackage[english]{babel}
\usepackage[T1]{fontenc}
\usepackage{amsmath}
\usepackage{amsfonts}
\usepackage{amssymb}
\usepackage{amsthm}
\usepackage{hyperref}

\usepackage{tikz}
\usepackage{lipsum}

\newtheorem{theorem}{Theorem}

\begin{document}

\title{Amplification of quadratic Hamiltonians}

\author{Christian Arenz} 
\affiliation{Frick Laboratory, Princeton University, Princeton NJ 08544, United States}

\author{Denys I. Bondar} 
\affiliation{Tulane University,  New Orleans, LA 70118, United States}

\author{Daniel Burgarth} 
\affiliation{Department of Physics and Astronomy, Macquarie University, Sydney, NSW 2109, Australia}

\author{Cecilia Cormick} 
\affiliation{Instituto de F\'{i}sica Enrique Gaviola, CONICET and Universidad Nacional de C\'{o}rdoba,
Ciudad Universitaria, X5016LAE, C\'{o}rdoba, Argentina}

\author{Herschel Rabitz} 
\affiliation{Frick Laboratory, Princeton University, Princeton NJ 08544, United States}

\maketitle

\begin{abstract}
Speeding up the dynamics of a quantum system is of paramount importance for 
quantum technologies. 
However, in finite dimensions and without full knowledge of the details of the system, it is 
easily shown to be \emph{impossible}. In contrast we show that 
continuous variable systems described by a certain class of quadratic 
Hamiltonians can be sped up without such detailed knowledge. We call the 
resultant procedure \emph{Hamiltonian amplification} (HA). The HA method relies on the 
application of local squeezing operations allowing for amplifying even unknown or noisy couplings 
and frequencies by acting on individual modes. Furthermore, we show how to combine HA  
with dynamical decoupling to achieve amplified Hamiltonians that are free from 
environmental noise. Finally, we illustrate a significant reduction in gate times of cavity 
resonator qubits as one potential use of HA.  
\end{abstract}

\section{Introduction}

Strong interactions between the components of a quantum device are crucial for 
maintaining the quantum effects relevant for quantum information processing. For 
instance, the generation of entanglement between system components 
\cite{Entanglement1, Entanglement2}, the implementation of multi-qubit gates 
\cite{MultiQGates1, MultiQGates2, MultiQGates3}, and, in general,  coherent 
interactions used in quantum metrology and quantum sensors \cite{Sensing1, 
Sensing2}, all rely on sufficiently strong couplings. Furthermore, the 
speed at which a quantum system can evolve towards a desired target is inherently 
limited by the strength of the underlying Hamiltonian which governs the 
evolution \cite{SpeedLimitsRev1, SpeedLimitsRev2}. A generic framework 
describing how such couplings/processes can be amplified is therefore highly 
desirable. As schematically summarized in Fig. 1, while \emph{Dynamical Decoupling} (DD)  achieves the opposite by suppressing unknown couplings/processes through rapidly applied controls \cite{LVioal1,Dec1,BookLidar}, here we propose an amplification procedure referred to as \emph{Hamiltonian 
amplification} (HA). In particular, we show that through rapidly applying parametric controls, even unknown parameters of an important class of quadratic Hamiltonians (see.
Eq. \eqref{eq:hamquad}) can generically be enhanced. Consequently, by combining HA and (partial) DD, system parameters can selectively be amplified and unwanted interactions, for instance with the environment, are simultaneously suppressed, thereby opening up a path for the manipulation of quantum systems in ways that were commonly thought to be 
impossible.

\begin{figure}[t]
\centering
\includegraphics[width=0.75\columnwidth]{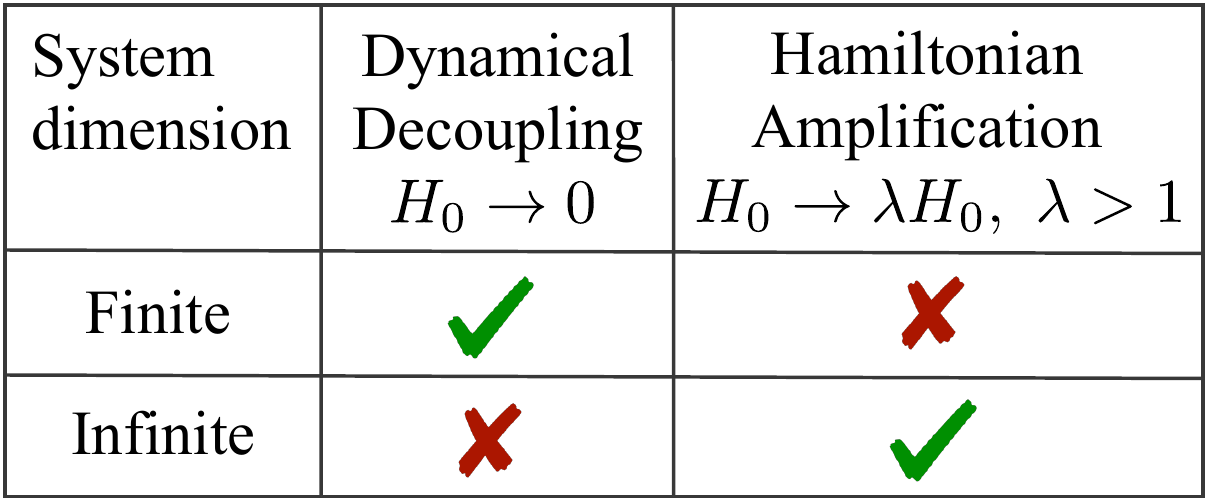}\caption{
\label{fig:intro}  The paper introduces Hamiltonian Amplification (HA) that allows for amplifying an unknown Hamiltonian $H_{0}$ by an amount $\lambda>1$.  While HA is impossible in finite dimensional systems, we show that some infinite dimensional system (Eq. \eqref{eq:hamquad}) can be amplified through rapidly applied local parametric controls.  In contrast, for some infinite dimensional systems the opposite, i.e., Dynamical Decoupling (DD) that achieves averaging out Hamiltonians, is impossible \cite{Dec1}, while for finite dimensional systems DD can always succeed \cite{BookLidar}. However, for infinite dimensional systems DD can be used to suppress certain interactions. Thus, we show that HA can be combined with DD to enhance desired evolutions and simultaneously fight decoherence.}
\end{figure}

Given a time independent Hamiltonian $H_0$, our goal is to turn the 
natural evolution $\exp (-iH_0t)$ into an accelerated evolution 
$\exp (-i\lambda H_0 t)$,  $\lambda > 1$, by adding a suitable control Hamiltonian 
$H_{c}(t)$ that induces a speed-up of the  system's evolution. If we 
know $H_0$ and have full control over the system, this can be trivially 
achieved by choosing $H_c(t)=H_c=(\lambda-1)H_0$. But the realistic case of 
interest arises if some parameters of the Hamiltonian are 
unknown or if we have only partial control of the system. If the speed-up can be 
achieved under these conditions, a physical effect observed at a time $t$ can 
then be observed at a shorter time $t/\lambda$ without knowing the details of 
the system. For example, a particularly important goal would be amplification of 
the (possibly uncertain) interaction strength between 
qubits by local controls only. 

 While in finite dimensions such
amplification is generally \emph{not possible} with either unknown parameters or 
partial control (as a consequence of norm preservation \cite{Lidar}) we show, surprisingly, 
that \emph{both} limitations can 
be overcome in some infinite-dimensional systems by parametrically driving the 
system components. 

Unlike other protocols recently developed for superconducting circuits and opto-mechanics \cite{Sim1, Sim2, Sim3} that suffer from  also increasing the interactions with the environment \cite{Home}, HA provides here a general framework that goes beyond a particular implementation, and combined with DD can avoid this drawback.

\section{Amplifying quadratic Hamiltonians}
Consider a quantum system described by a Hamiltonian $H_{0}$ that 
is driven by means of an external, possibly time dependent, control Hamiltonian 
$H_{c}(t)$ so that the total Hamiltonian governing the system's evolution reads $H(t)=H_{0}+H_{c}(t)$. Analogous to DD \cite{BookLidar}, in order to introduce HA we separate the total evolution into that of the 
controller alone $U_{c}(t)$ and its action on $H_{0}$. The total evolution is 
then given by $U(t)=U_{c}(t)\tilde{U}(t)$ where the evolution in the frame 
rotating with $U_{c}(t)$ reads $\tilde{U}(t)= \exp(-it 
\bar{H}(t))$ where 
$\bar{H}(t)$ is given by the Magnus expansion \cite{BookLidar,Magnus}. The goal of HA is to find controls 
so that $\bar{H}(t)=\lambda H_{0}$, $\lambda>1$, for all $t$. As proven in Appendix \ref{app1}, this is impossible for finite 
dimensional systems. Note, however, that DD, i.e., $\lambda=0$ for traceless $H_{0}$, can always be achieved when the system is finite dimensional \cite{LVioal1}. The situation changes when continuous variable systems, such as quantum harmonic oscillators, living in an infinite dimensional space, are considered. Here DD \emph{does not} always work \cite{Dec1}, while, as we show below, HA is possible.   

 In order to introduce the concept, we start by 
considering the evolution of a quantum harmonic oscillator with frequency 
$\omega$ described by the Hamiltonian $H_{0} 
=\frac{\omega}{2}(x^{2}+p^{2})$ where $x$ and $p$ are the canonical position and 
momentum operators, respectively. We begin by considering bang-bang type controls that correspond to delta-function like pulses implementing, at time intervals $\Delta t=\frac{t}{2n}$,   
the squeezing operations $S^{(\pm)}=\exp[\pm i 
\frac{r}{2}(xp+px)]$ where $r$ is the squeezing parameter. The canonical 
operators transform 
under squeezing according to $S^{(\pm)\dagger}xS^{(\pm)}=\exp(\mp 
r)x$ and $S^{(\pm)\dagger}pS^{(\pm)}=\exp(\pm r)p$. 
Thus, if we alternate between $S^{(+)}$ and $S^{(-)}$ the modified dynamics is 
approximately given by
 $ \left(S^{(-)\dagger}e^{-i\frac{H_{0}t}{2n}}S^{(-)}S^{(+)\dagger}e^{-i\frac{H_{0}t}{2n}}S^{(+)}\right)^{n}\approx e^{-i\cosh(2r)H_{0} t}$, which becomes exact in the limit of rapid squeezing, i.e, $\Delta t\to 0$, 
$n\to\infty$ with $t$ fixed.
Consequently the operations $V=\{S^{(\pm)}\}$ amplify the harmonic 
oscillator by a factor $\lambda=\cosh(2r)$. The physics behind the amplification may be readily understood: squeezing in 
the $x$-direction maps the operator $x$ onto $e^{-r}x$, while squeezing in 
$p$-direction maps $x$ onto $e^r x$. In the case of very fast alternation, the 
operator is effectively averaged into $\cosh(r)x$ and thus amplified by a 
factor $\cosh(r)$, and the same applies to any other phase-space 
quadrature. For the operators $x^2$ and $p^2$, the same reasoning leads to an 
amplification by $\cosh(2r)$. We now
generalize this idea to more complex continuous-variable systems. 
 
Consider again bang-bang controls that instantaneously implement 
a set of (Gaussian) unitary operations $V$ in a Suzuki-Trotter 
type 
sequence $\Lambda_{t/n}=\prod_{v\in V}v^{\dagger}\exp(-iH_{0}\frac{t}{|V|n})v$, 
where $H_0$ is a quadratic Hamiltonian \cite{QuadraticH, Serafini} of a system of continuous variables (see 
below for details) and we denote by $|V|$ the number of elements in $V$. 
Since quadratic 
Hamiltonians can be represented by matrices, the map $\Lambda_{t/n}$ converges 
and in the limit we have $\lim_{n\to\infty}\Lambda_{t/n}^{n}=\exp(-it M(H_{0}))$ where the 
dynamics  is governed by the average \cite{BookLidar},
\begin{align}
\label{eq:ampmap}
M(H_{0})=\frac{1}{|V|}\sum_{v\in V}v^{\dagger}H_{0}v. 	
\end{align}

A most natural class of continuous-variable systems to consider is described 
by quadratic Hamiltonians of the form 
 \begin{align}
 \label{eq:hamquad}
H_{0}=\sum_{i,j}(\omega_{i,j}^{(x)}x_{i}x_{j}+\omega_{i,j}^{(p)}p_{i}p_{j}),	
\end{align}
which typically serve as a model for a wide range of quantum-optical and 
opto-mechanical systems \cite{QuadraticH1} consisting of $N$ linearly 
interacting quantum harmonic oscillators, where $\omega_{i,j}^{(x)}$ and 
$\omega_{i,j}^{(p)}$ determine the frequencies 
$(i=j)$ and coupling constants $(i\neq 
j)$. Now, assume that the dynamics can be modified through 
$V=\{\mathcal{S}^{(\pm)}\}$ where 	$\mathcal{S}^{(\pm)}=\prod_{i=1}^{N}S_{i}^{(\pm)}$,
is a product of local squeezing operations $S_{i}^{(\pm)}$ in $x_{i}$ and 
$p_{i}$ direction of the $i$th oscillator.  We find $M(H_{0})=\cosh(2r)H_{0}$ for all Hamiltonians of the form \eqref{eq:hamquad}, concluding that through local squeezing operations any quadratic 
Hamiltonian \eqref{eq:hamquad} can be amplified by a factor 
$\cosh(2r)$ without knowledge of the frequencies and coupling constants present 
in $H_{0}$. Thus, contrary to the case of
finite-dimensional systems, coupling strengths and frequencies
can be amplified through local unitary operations. We remark here that for two linearly 
interacting harmonic oscillators the coupling can also be enhanced by acting on a single 
oscillator alone. This could be of 
particular importance for amplifying the entanglement creation between a light 
mode and a mechanical oscillator in an optomechanical system 
\cite{Entanglement2, OptoMRev}.

The quadratic Hamiltonian \eqref{eq:hamquad} differs 
from a generic quadratic Hamiltonian $H=\frac{1}{2}\sum_{i,j}A_{i,j}R_{i}R_{j}$ where $A$ is a real and symmetric 
$2N\times 2N$ matrix and the canonical operators constitute the vector 
$R=(x_{1},p_{1},\cdots, x_{N},p_{N})$ \cite{Serafini}, only by single-mode and two-mode squeezing 
terms $\propto x_{i}p_{j}$, which we collect to form the Hamiltonian $H_{1}$. Such terms are left invariant by squeezing the oscillator in $x$ and $p$ directions
and thus, a generic quadratic Hamiltonian $H=H_{0}+H_{1}$ is transformed under amplification through $\{\mathcal S^{(\pm)}\}$ according to $M(H)=\cosh(2r)H_{0}+H_{1}$. We see that if the actual system under consideration deviates from the quadratic Hamiltonian \eqref{eq:hamquad} by the presence 
of two-mode squeezing terms, such terms become 
negligible when the system is sufficiently amplified. We further remark that two-mode squeezing terms can be amplified by locally squeezing around different angles. However, we did not find a procedure yet that amplifies generic quadratic Hamiltonians.

The HA procedure above assumed that the squeezing operations were  
implemented instantaneously and alternating infinitely fast. The performance of 
the protocol is now considered under more 
realistic conditions. The error $\epsilon$ that is induced for finite $n$, i.e. for a finite 
waiting time $\Delta t$, can be upper-bounded by using well-known bounds for 
the Trotter sequence \cite{Trotter} and exploiting the symplectic representation 
of 
quadratic Hamiltonians \cite{QuadraticH, Serafini}. That is, the time evolution 
$U=\exp(-iH_{0}t)$ of a quadratic Hamiltonian $H_{0}$ is represented by a 
symplectic transformation $\exp(-tA\Omega)\in \text{Sp}(2N,\mathbb{R})$ where 
$\Omega$ is the symplectic form. Using the symplectic representation of 
$\Lambda_{t/n}$ and $M(H_{0})$, with further details found in Appendix \ref{app2}, we can upper-bound the error by  
\begin{align}
\label{eq:error}
\epsilon\leq \frac{t\Delta t \omega_{\text{max}}^{2} 
N^{2}}{4}|\sinh(4r)|e^{t\omega_{\text{max}}N\sqrt{\frac{\cosh(4r)}{2}}},
\end{align}
 where the error was evaluated using the Hilbert-Schmidt norm difference of the 
symplectic matrices and 
$\omega_{\text{max}}=\max_{i,j}\{|\omega_{i,j}^{(x)}|,|\omega_{i,j}^{(p)}|\}$ 
is 
the highest frequency/coupling constant present in the amplified Hamiltonian\footnote{While the Hilbert-Schmidt norm difference $\epsilon$ of symplectic matrices is typically used as a distance measure for (quadratic) continuous variable control systems (see e.g.,  \cite{HerschContinuous, DanielContinuous}), we remark here that it is not an operational fidelity measure. However, in the spirit of \cite{Silvan}, $\epsilon$ can be made operational if one imposes energy constraints on the system Hamiltonian $H$ by, for instance, considering only Gaussian states for which the expectation of $H$ is upper bounded by a predefined value.}. While HA is about transforming Hamiltonians rather than preparing states, in Appendix \ref{app4} we also investigate an example of the ``Gaussian'' fidelity error \cite{Banchi} to speed up the preparation of Gaussian states using HA. 

Instead of using 
instantaneous bang-bang operations, we now turn to  
amplification using smooth pulses with a finite amplitude. A time-dependent squeezing parameter $r(t)$ is considered such that the 
Hamiltonian of the controller reads 
$H_{c}(t)=r(t)\sum_{j=1}^{N}(x_{j}p_{j}+p_{j}x_{j})$. We take a 
periodic controller $U_{c}(t+2n\Delta t)=U_{c}(t)$ for $n\in\mathbb N$ so that 
the time evolution at times $t=n2\Delta t$ in the presence of $H_{0}$ is given by the 
unitary transformation $U(t)=\exp(-2in\Delta t\bar{H})$. Instead of modifying the dynamics with instantaneous unitary operations 
$\mathcal S^{(\pm)}$ applied in cycles of duration $2\Delta t$, we now change the 
dynamics smoothly, such that for small $\Delta t$ the dynamics is effectively 
governed by the 1st order term of the Magnus expansion $\bar{H}^{(0)}=\frac{1}{2\Delta t}\int_{0}^{2\Delta t}U_{c}^{\dagger}(t)H_{0}U_{c}(t)dt$, which is the 
continuous equivalent to the map \eqref{eq:ampmap}. The higher-order terms of 
the Magnus expansion introduce errors which can be neglected if $\Delta t$ is 
sufficiently small. As an example, with details found in  Appendix \ref{app3}, 
the pulse
$r(t)=\frac{\pi K}{2\Delta t}\cos\left(\frac{\pi t}{\Delta t}\right)$ yields $\bar{H}^{(0)}=\lambda H_{0}$ with $\lambda=I_{0}(K)$ being the modified Bessel function of the first kind. Here, the amount of amplification is 
chosen by the constant $K$, which, together with the frequency of the pulse 
$\propto \Delta t^{-1}$, determines the pulse amplitude. We remark that for 
$\Delta t\to 0$, which yields perfect amplification, the pulse 
amplitude as well as its frequency becomes infinitely large, thus 
approaching an infinitely fast bang-bang type control. However, we further show in Appendix \ref{app3} that a large class of pulses exist that amplify the 1st order and simultaneously suppress the 2nd order of the Magnus expansion. Such pulses can be more effective in achieving HA, which was confirmed by numerical simulations. Furthermore, a detailed analysis in Appendix \ref{app4} of the efficiency of imperfect controls upon amplification indicates that errors in the amplification pulses enter for sufficiently small $\Delta t$ linearly in the amplification process.

\section{Amplification and Decoupling}
While the system 
parameters are amplified, possibly unwanted couplings with  
the environment are amplified too. However, we now show how HA can be combined with some form of DD so that interactions with the environment are suppressed 
while system parameters are amplified, thereby achieving amplification that is free from environmental noise.

As shown in \cite{Dec1}, it is not possible to average generic quadratic Hamiltonians to zero (or to multiples of the identity) through rapidly applied unitary operations (i.e., DD in the sense introduced before). However, parts of the Hamiltonian, such as interactions quadratic in the quadrature operators, can be suppressed \cite{Dec1, Vitali}. In general, in the limit of infinitely fast operations the 
dynamics is governed by a map $M^{\text{(\text{DD})}}$ of the form 
\eqref{eq:ampmap}, where the decoupling set $V^{(\text{DD})}$ consists of  
unitary transformations that allow for suppressing the unwanted interactions. Consider a continuous-variable system $S$ that interacts with 
another continuous-variable system $E$ in a linear way described by the total 
Hamiltonian $H=H_{S}+H_{E}+H_{SE}$ where $H_{S},~H_{E}$ are the Hamiltonians of 
$S$ and $E$, respectively, and $H_{SE}=\sum_{i\in S,j\in E}(\omega_{i,j}^{(x)}x_{i}x_{j}+\omega_{i,j}^{(p)}p_{i}p_{j})$ describes the interaction, assuming that 
all Hamiltonians are of the form \eqref{eq:hamquad}. The dynamics of system $S$ 
can always be decoupled from $E$ using the operations 
$V_{S}^{(\text{DD})}=\{\mathcal R_{S}(0),\mathcal R_{S}(\pi)\}$ where $\mathcal 
R_{S}(\phi)=\prod_{i\in S}\exp(-i\frac{\phi}{2}(x_{i}^{2}+p_{i}^{2}))$ is a 
product of local rotations of the system oscillators \cite{Dec1, Vitali}. The 
operation $\mathcal R_{S}(\pi)$ inverts the sign in front of $H_{SE}$ so that 
the interaction between $S$ and $E$ is suppressed by acting on $S$ alone, i.e., 
$M_{S}^{(\text{DD})}(H)=H_{S}+H_{E}$. In contrast to  finite-dimensional 
systems, the dynamics of $S$ remains invariant under the decoupling map 
$M_{S}^{(\text{DD})}$, which is the attractive feature of why DD and 
HA can be combined. 

Let $M_{S}^{(\text{HA})}$ be the map that amplifies system $S$ by a factor $\cosh(2r)$ through the set of operation $V_{S}^{(\text{HA})}=\{\mathcal S_{S}^{(\pm)}\}$ where $\mathcal S_{S}^{(\pm)}=\prod_{i\in S}S_{i}^{(\pm)}$ so that $M_{S}^{({\text{HA}})}(H_{S})=\cosh(2r)H_{S}$, noting that any linear interaction with $E$ is amplified too, i.e, $M_{S}^{{(\text{HA}})}(H_{SE})=\cosh(r)H_{SE}$. However, through concatenating both maps $\mathcal G_{S}\equiv M_{S}^{(\text{HA})}\circ M_{S}^{(\text{DD})}$ we see that $\mathcal G_{S}(H)=\cosh(2r)H_{S}+H_{E}$. Thus, system $S$ becomes simultaneously amplified and decoupled from the environment $E$. The map $\mathcal G_{S}(\cdot)=\frac{1}{|G_{S}|}\sum_{g\in G_{S}}g^{\dagger}(\cdot)g$ that achieves amplification and decoupling is determined by the set of operations $G_{S}=\{\mathcal R_{S}(\pi)\mathcal S_{S}^{(+)},\mathcal S_{S}^{(+)},\mathcal R_{S}(\pi)\mathcal S_{S}^{(-)},\mathcal S_{S}^{(-)}\}$
applied only on $S$. Thus, through alternating between squeezing and rotating the system oscillators, the interaction with another system is suppressed and the system of interest is amplified, noting that DD can also be achieved using smooth pulses \cite{Dec1}.

\section{Amplification of qubit interactions} 
The proposed HA only works in an infinite-dimensional setting. However, 
finite-dimensional systems, such as qubits, can be coupled through the 
interaction with 
a quantum harmonic oscillator such that an effective qubit-qubit interaction is obtained \cite{CircuitQED1, CircuitQED2, CircuitQED, Ions, 
Int1, Retzker}. If the original couplings between the qubits and the oscillator can 
be amplified, so can the effective coupling strength between the qubits. 
For instance, for circuit QED systems \cite{CircuitQED1, CircuitQED2, 
CircuitQED} it was recently shown that the interaction strength between a charge 
qubit and a resonator mode is enhanced in the regime where the resonator mode is 
strongly squeezed \cite{Sim1}. However, strongly squeezing the resonator mode is experimentally challenging. In contrast, using HA to amplify the qubit-resonator coupling has the advantage that operating in the strong squeezing regime is not required. Instead, the qubit-resonator coupling is  enhanced by rapidly alternating between squeezing the resonator mode along different directions.

The interaction between a resonator mode and two charge qubits can be described by 
a Jaynes-Cummings type Hamiltonian 
$H_{0}=H_{\text{r}}+H_{q}+H_{\text{int}}$ where $H_{r}=\omega_{r} a^{\dagger} 
a$ and $H_{\text{q}}=\sum_j \frac{\omega_{j}}{2}\sigma_{z}^{(j)}$ are the free 
Hamiltonians of the resonator mode and the qubits $j=1,2$, respectively, and 
$H_{\text{int}}=\sum_{j}g_{j}(\sigma_{+}^{(j)}a+\sigma_{-}^{(j)}a^{\dagger})$ 
describes the interaction. If one assumes that the oscillator is initially 
prepared close to the vacuum state and the qubits are strongly detuned from the 
resonator frequency, $\Delta_{j}=|\omega_{j}-\omega_{r}|\gg g_{j}$, the 
oscillator degrees of freedom can be eliminated so that between the qubits an 
effective exchange interaction 
$V=\Omega(\sigma_{-}^{(1)}\sigma_{+}^{(2)}+\sigma_{+}^{(1)}\sigma_{-}^{(2)})$ is 
induced. For simplicity, we assume that $\omega=\omega_{j}$ and 
$g=g_{j}$ for $j=1,2$, so that the interaction strength is given by 
$\Omega=\frac{g^{2}}{\Delta}$. For $t=\frac{\pi}{2\Omega}$ such a Hamiltonian 
implements a SWAP gate. Since the bosonic annihilation operator transforms under the squeezing operation $S^{(\pm)}$ according 
to $S^{(\pm)\dagger}aS^{(\pm)}=a\cosh(r)\mp a^{\dagger}\sinh(r)$ we find for 
the Jaynes Cummings Hamiltonian 
$M(H_{0})=H_{\text{q}}+\cosh(2r)H_{\text{r}}+\cosh(r)H_{\text{int}}$\,. Thus, in 
the limit of squeezing the oscillator infinitely fast, the qubit-qubit 
interaction is determined by the amplified interaction strength
\begin{align} 
\label{eq:ampfrequency}
\Omega_{\text{amp}}(r)= \frac{\cosh^{2}(r)g^{2}}{|\omega-\cosh(2r)\omega_{r}|},
\end{align} 
so that the time to implement a SWAP gate is reduced to 
$t_{\text{swap}}(r)=\frac{\pi}{2\Omega_{\text{amp}}(r)}$. However, the effective 
qubit-qubit interaction strength cannot be arbitrarily enhanced, since for large squeezing parameters $r\gg 1$ the amplified strength 
$\Omega_{\text{amp}}$ tends to $\frac{g^{2}}{2\omega_{r}}$, as a consequence 
of the effect of the squeezing on the oscillator frequency as well as on the 
couplings\footnote{We remark here that in the regime in which the resonator mode is rapidely squeezed along $x$ and $p$ such that the dynamics is effectively governed by $M(H_{0})$, the initial state of the resonator mode is not changed. Thus, as long as the condition $|\omega-\cosh(2r)\omega_{r}|\gg \cosh(r)g$ is satisfied to eliminate the resonator mode in the presence of HA, then independently of how strong the resonator is squeezed an effective two qubit interaction determined by \eqref{eq:ampfrequency} is obtained.}. However, numerical investigations indicate that for finite squeezing 
a significant speed-up can be achieved.

Based on 
the full Jaynes-Cummings Hamiltonian, we numerically studied excitation 
transfer from one qubit to another, using bang-bang amplification. The 
qubits were initially prepared in a product state consisting of one up and one 
down 
eigenstate of $\sigma_{z}$ and the colormap of Fig. \ref{fig:probability} shows 
the probability for swapping the qubit states as a function of the evolution 
time $t$ and the squeezing parameter $r$. The 
parameters can be found in the caption of 
Fig. \ref{fig:probability}.
\begin{figure}
\centering
\includegraphics[width=0.80\columnwidth]{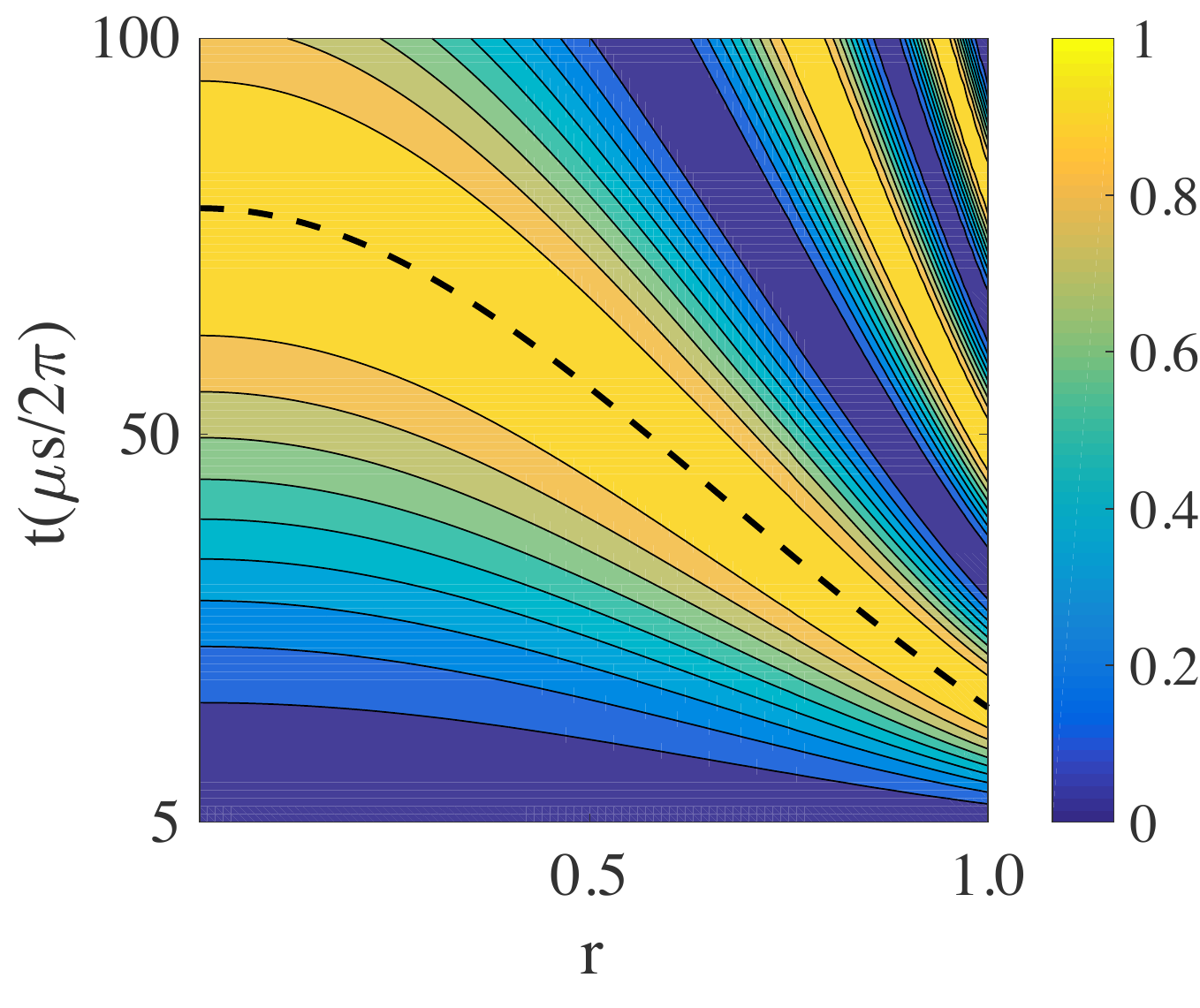}\caption{
\label{fig:probability} Probability of swapping the states of two qubits 
(colormap) that interact via a quantum harmonic oscillator as a function of the 
amount $r$ the oscillator is rapidly squeezed and the total evolution time $t$. 
The time to swap the state (dashed black curve) is determined by the amplified 
frequency \eqref{eq:ampfrequency} with parameters 
$\omega_{\text{r}}/2\pi=2.5\,\text{GHz},~\omega/2\pi=15\,\text{GHz},
~g/2\pi=50\,\text{MHz},$ taken in the range of a typical circuit QED setup 
\cite{CircuitQED}.}
\end{figure}
We see that the oscillations are accelerated by increasing the amount of 
squeezing. According to $t_{\text{swap}}(r)$ shown as the black dashed line, the 
time to swap the qubit states is significantly reduced. Unfortunately, in 
order to achieve such an improvement for the chosen parameters $\Delta t \approx 
1\text{ps}$. However, the better performance for the case when bounded smooth controls are used suggests that more sophisticated shaped pulses could lead to an efficient amplification of qubit interactions. Moreover, by combining HA with DD, obtained by additionally applying $\mathcal R(\pi)$ on the oscillator and $\sigma_{z}$ on the qubits such that the Jaynes-Cummings interactions stays invariant while interactions with the environment are suppressed, gates can be implemented faster and at the same time decoherence is reduced.  
     
\section{Conclusions}
As indicated in Fig. 1, every (traceless) Hamiltonian describing a finite-dimensional quantum system
can be averaged to zero using a decoupling 
sequence. However, in finite dimensions the Hamiltonian cannot be amplified. In 
contrast, some infinite-dimensional systems cannot be averaged out, 
whereas, as we have shown, some infinite-dimensional systems can be amplified. 
This observation has formed the basis of 
\emph{Hamiltonian Amplification} introduced here and proven to be applicable to a 
broad class of continuouts-variable systems through local squeezing operations 
even without full knowledge of, or full control over, the parameters present in the 
Hamiltonian. We showed how Hamiltonian Amplification can be combined with Dynamical Decoupling, which allows for simultaneously amplifying the system parameters and suppressing unwanted interactions. We further illustrated Hamiltonian Amplification in a hybrid system to speed up the implementation of quantum logic gates, thereby showing the broad scope of the proposed scheme in manipulating time scales of quantum systems. 
By amplifying/suppressing the relevant coupling and frequency components, we believe that the 
scheme opens up new prospects for a wide range of areas \cite{Sensing2,Entanglement2,OptoMRev,Romero,Hensinger,Networks,HerschMe}. We further note that the simplest form of HA was recently demonstrated in a trapped-ion experiment by amplifying the position operator through alternating between squeezing and anti-squeezing, indicating that HA is feasible with current technology \cite{Burd}.

There is certainly room for more sophisticated pulse sequences, possibly obtained from 
well-known methods used in dynamical decoupling \cite{Euler, Robust, HigerOrder, UhrigBos} 
and optimal control theory \cite{ControlRev1}, that would allow for even more 
efficient and robust amplification. Furthermore, in this work Hamiltonian Amplification was introduced for quadratic Hamlitonians but the observation that $M^{(\text{HA})}(x^{n}+p^{n})=\cosh(nr)(x^{n}+p^{n})$ with $n$ being an integer suggest that some non-linear terms can be amplified too, which will be investigated in future work.

\section*{Acknowledgements}
 C.C. acknowledges funding from the grant PICT 
2015-2236. D.I.B. is supported by AFOSR (grant FA9550-16-1-0254), ARO (grant W911NF-19-1-0377), and DARPA (grant D19AP00043). C. A. was supported by the NSF (grant CHE-1464569) and the ARO (grant W911NF-19-1-0382). H. R. was supported by the DOE (grant DE-FG02-02ER15344). 

\noindent 
\bibliographystyle{unsrturl}
\bibliography{source05132020.bbl}

\onecolumn
\appendix

\section*{Appendix}

\section{No amplification of bounded Hamiltonians}\label{app1}
The proof that finite dimensional quantum systems cannot be amplified follows directly from results obtained in \cite{Lidar}. We consider a finite dimensional quantum system whose total Hamiltonian is given by $H(t)=H_{0}+H_{c}(t)$, where $H_{0}$ is the Hamiltonian that we want to amplify by a coherent controller described by the Hamiltonian $H_{c}(t)$. The total evolution $U(t)=U_{c}(t)\tilde{U}(t)$ can be separated into the evolution of the controller alone $U_{c}(t)=\mathcal T\exp(-i\int_{0}^{t}H_{c}(t^{\prime})dt^{\prime})$ and its action on $H_{0}$ given by $\tilde{U}(t)=\mathcal T\exp(-i\int_{0}^{t}U_{c}^{\dagger}(t^{\prime})H_{0}U_{c}(t^{\prime})dt^{\prime})=\exp(-it\bar{H}(t))$ where $\bar{H}(t)$ is given by the Magnus expansion. Since there always exist a set of unitary transformations $\{W(s)\}$ that allows to express $\bar{H}(t)$ as $\bar{H}(t)=\frac{1}{t}\int_{0}^{t}W^{\dagger}(t^{\prime})U_{c}^{\dagger}(t^{\prime})H_{0}U_{c}(t^{\prime})W(t^{\prime})dt^{\prime}$, we have that for any unitarily invariant norm $\Vert\bar{H}(t)\Vert\leq \Vert H_{0}\Vert$ for all $t$, which implies that Hamiltonian Amplification is not possible for finite dimensional systems. Note, however, that the proof becomes meaningless when infinite dimensional systems described by unbounded operators are considered, as then the norms are not defined.  

\section{Derivation of the error bound}\label{app2}
In order to bound the error for obtaining an evolution that is generated by the amplified Hamiltonian $M(H_{0})=\frac{1}{2}(\mathcal S^{(+)\dagger}H_{0}\mathcal S^{(+)}+\mathcal S^{(-)\dagger}H_{0}\mathcal S^{(-)})=\cosh(2r)H_{0}$, valid for any quadratic Hamiltonian of the form given in Eq. (2) of the main body of the manuscript, we use the symplectic representation of quadratic Hamiltonians. That is, the unitary time evolution $U=e^{-iH_{0}t}$ is represented by a symplectic matrix $e^{-A_{0}\Omega t}\in \text{Sp}(2N,\mathbb R)$  where $A_{0}\in\mathbb R^{2N\times 2N}$ is a symmetric matrix and $\Omega$ is the symplectic form. If we denote by $A^{(\pm)}$ the to $\mathcal S^{(\pm)^{\dagger}}H_{0}\mathcal S^{(\pm)}$ corresponding real and symmetric matrices such that $\frac{1}{2}(A^{(+)}+A^{(-)})=\cosh(2r)A_{0}$, the error $\epsilon=\left\Vert e^{-\cosh(2r)A_{0}\Omega t}-e^{-\frac{t}{2n}A^{(+)}\Omega}e^{-\frac{t}{2n}A^{(-)}\Omega} \right \Vert$ for obtaining an amplified $A_{0}$ is upper bounded by \cite{Trotter}, 
\begin{align}
\epsilon \leq \frac{t\Delta t}{8}\Vert [A^{(+)}\Omega,A^{(-)}\Omega]\Vert e^{\frac{t}{2}(\Vert A^{(+)}\Omega\Vert +\Vert A^{(-)}\Omega\Vert )},
\end{align}
 where $\Delta t=\frac{t}{n}$ and $\Vert A\Vert=\sqrt{\text{tr}\{A^{\dagger}A\}}$ is the Hilbert-Schmidt norm. Since $\Vert [A^{(+)}\Omega,A^{(-)}\Omega]\Vert \leq  2 |\sinh(4r)| N^{2}\max_{i,j}|A_{i,j}|^{2}$,
and $\Vert A^{(\pm)}\Omega\Vert\leq \sqrt{2\cosh(4r)}N\max_{i,j}|A_{i,j}|$, with $\omega_{\text{max}}=\max_{i,j}|A_{i,j}|$ we arrive at the desired bound \eqref{eq:error}. 

\section{1st and 2nd order terms of the Magnus expansion}\label{app3}
Here we derive the pulse given in the main body of the manuscript as well as show that a large class of pulses 
simultaneously amplifies the 1st order and suppresses the 2nd order terms of the 
Magnus expansion. We recall from the main body of the paper that the 
first order of the Magnus expansion is given by 
\begin{align}
\bar{H}^{(0)}=\frac{1}{2\Delta t}\int_{0}^{2 \Delta t}H(t)dt,
\end{align}
where
\begin{align}
H(t)&=U_{c}^{\dagger}(t)H_{0}U_{c}(t) \nonumber \\
&=\sum_{i,j}\left(\omega_{i,j}^{(x)}x_{i}x_{j}e^{-2R(t)}+\omega_{i,j}^{(p)}p_{i}p_{j}e^{2R(t)}\right),
\end{align}
with $R(t)=\int_{0}^{t}r(t^{\prime})dt^{\prime}$ being the integrated pulses. Thus, in order to amplify the 1st order of the Magnus expansion we need to find a pulse $r(t)$ that yields $\bar{H}^{(0)}=\lambda H_{0}$ where $\lambda=\frac{1}{2\Delta t}\int_{0}^{2\Delta t}\exp(\pm 2R(t))dt>1$. Since $\int_{0}^{2\Delta t}\exp(\pm K\sin(\pi t/\Delta t))dt=2\Delta t I_{0}(K)$ where $I_{0}(K)$ is the modified Bessel function of the first kind, we find that the pulse 
\begin{align}
r(t)=\frac{\pi K}{2\Delta t}\cos\left(\frac{\pi t}{\Delta t}\right),
\end{align} 
amplifies the 1st order of the Magnus expansion by an amount $\lambda=I_{0}(K)$. 

We now turn to the question whether there exist pulses that amplify the 1st order and simultaneously suppress the 2nd order of the Magnus expansion. 
The 2nd order terms of the Magnus expansion reads 
\begin{align}
\bar{H}^{(1)}&=-\frac{i}{4\Delta t}\int_{0}^{2\Delta t}dt_{1}\int_{0}^{t_{1}}dt_{2}[H(t_{1}),H(t_{2})]\nonumber \\
	&=-\frac{i}{2\Delta t}\sum_{i,j,i^{\prime},j^{\prime}}\omega_{i,j}^{(x)}\omega_{i^{\prime},j^{\prime}}^{(p)}[x_{i}x_{j},p_{i^{\prime}}p_{j^{\prime}}]\int_{0}^{2\Delta t}dt_{1}\int_{0}^{t_{1}}dt_{2}\sinh(2(R(t_{2})-R(t_{1}))).
\end{align}
Thus, the integrated pulse $R(t)$ that simultaneously amplifies $\bar{H}^{(0)}$ and suppresses $\bar{H}^{(1)}$ has to simultaneously satisfy the two conditions 
\begin{align}
\label{eq:I1}
	I_{1} \equiv& \int_{0}^{T}\sinh(u(t))dt=0,\\ 
	I_{2}\equiv& \int_{0}^{T}dt_{1}\int_{0}^{t_{1}}dt_{2}\sinh(u(t_{1})-u(t_{2}))=0, \label{eq:I2}
\end{align}
where from now on we use the short hand notation $T=2\Delta t$ and $u(t)=2R(t)$.
\\
\\
\begin{theorem}\label{Therem1}
A function $u(x)$ such that $u(-x) = u(x+T)$ satisfies conditions \eqref{eq:I2}. In particular,
\begin{align}
	u(t) = \sum_{n=0}^{\infty} \left( a_n \sin\frac{(2n + 1)\pi t}{T} + b_n \cos\frac{2n \pi t}{T} \right).
\end{align}
\end{theorem}
\begin{proof}
Introducing new variables $t = (t_2 + t_1 - T)/2$, $\tau = (t_2 - t_1 + T)/2$, the integral \eqref{eq:I2} can be written as
\begin{align}
	I_2 &= 2 \int_0^{T/2} g(\tau) d\tau \int_{-\tau}^{\tau}  dt \sinh[u(t-\tau + T) - u(t + \tau)].   \label{IntG}
\end{align}
If the expression under the integral \eqref{IntG} is an odd function of $t$ (for all $\tau$), then $I_2 = 0$. This leads to the condition
\begin{align}\label{FuncEq1ForU}
	u(-t-\tau + T) - u(-t + \tau) = -u(t-\tau + T) + u(t + \tau), \quad \forall \tau.
\end{align}
A solution to the equation \eqref{FuncEq1ForU} satisfies Eq. \eqref{eq:I2}. Using new variables $x=t - \tau$ and $y = t + \tau$, Eq. \eqref{FuncEq1ForU} reads
\begin{align}\label{FuncEq1XY}
	u(-y+T) - u(y) = u(-x) - u(x+T), \quad \forall x, \, y.
\end{align}
As in the method of the separation of variables for PDEs, Eq. \eqref{FuncEq1XY} is broken into two separate equations:
\begin{align}
	& u(-y+T) - u(y) = \lambda, \quad \forall y, \label{FuncEqUy} \\ 
	& u(-x) - u(x+T) = \lambda, \quad \forall x,
\end{align}
where $\lambda$ is a separation constant. Substituting $y \to -x$ into Eq. \eqref{FuncEqUy} leads to
\begin{align} \notag
	u(-x) - u(x+T) = -\lambda , \qquad u(-x) - u(x+T) = \lambda, \quad \forall x.
\end{align}
Therefore, $\lambda = 0$ and we are left with the single linear equation:
\begin{align}\label{FinalSingleFuncEq}
	u(-x) = u(x+T).
\end{align}
The first set of linearly independent solutions of Eq. \eqref{FinalSingleFuncEq} can be found by substituting $u_n(t) = \cos(2n\pi t/T)$, $n=0,1,2,\ldots$; whereas, the second set reads $u_n(t) = \sin([2n+1]\pi t/T)$.
\end{proof}

Using this theorem, we show that Eq. \eqref{eq:I1} and Eq. \eqref{eq:I2} have the following solutions:
\newcommand*\widefbox[1]{\fbox{\hspace{2em}#1\hspace{2em}}}
\begin{align}
	u(t) &= \sum_{n=0}^{\infty} a_n \cos\frac{2\pi (2n + 1) t}{T}, \label{ExSol1} \\
	u(t) &= \left(\sin\frac{2\pi t}{T}\right)^2 \sum_{n=0}^{\infty} a_n \cos\frac{2\pi (2n + 1) t}{T}, \label{ExSol3} \\
	u(t) &= \sin\left(\frac{4\pi t}{T}\right) \sum_{n=0}^{\infty} a_n \sin\frac{2\pi (2n + 1) t}{T}, \label{ExSol4}  \\
	u(t) &= a_n \cos\frac{4\pi nt}{T}, \qquad n=1,2,3,\ldots. \label{ExSol2}
\end{align}

Clearly, the function \eqref{ExSol1} obeys Theorem \ref{Therem1} by construction. Now let us show that the function (\ref{ExSol1}) also satisfies the condition \eqref{eq:I1}:
\begin{align}
	I_1 &=   \int_0^T  \sinh\left(  \sum_{n=0}^{\infty} a_n \cos\frac{2\pi (2n + 1) t}{T} \right) dt 
		=	 \frac{T}{2\pi} \int_{0}^{2\pi}  \sinh\left(  \sum_{n=0}^{\infty} a_n \cos[(2n + 1)x] \right) dx. \notag
\end{align}
Recalling that $\cos(y \pm [2n+1]\pi) = -\cos(y)$, we obtain
\begin{align}
	& \int_{0}^{2\pi} \sinh\left(  \sum_{n=0}^{\infty} a_n \cos[(2n + 1)x] \right) dx 
	= \left( \int_{0}^{\pi} + \int_{\pi}^{\pi + \pi} \right) \cdots = \left( \int_{0}^{\pi} - \int_{0}^{\pi} \right) \cdots = 0 \Longrightarrow I_1 = 0.
\end{align}
In a similar fashion one can verify that the functions \eqref{ExSol3} and \eqref{ExSol4} satisfy the conditions \eqref{eq:I1} and \eqref{eq:I2}. The family of solutions \eqref{ExSol2} evidently obeys Theorem \ref{Therem1}. We show that it also satisfies \eqref{eq:I1},
\begin{align}
	I_1 &= \int_0^T \sinh \left(a_n \cos\frac{4\pi nt}{T}\right) dt = \frac{T}{4\pi n} \int_0^{4\pi n} \sinh (a_n \cos x) dx =  \frac{2n T}{4\pi n} \int_0^{2\pi} \sinh (a_n \cos x) dx \Longrightarrow \notag\\
	 & \int_0^{2\pi} \sinh (a_n \cos x) dx =  \left[ \int_0^{\pi/2} + \int_{\pi/2}^{\pi + \pi/2} + \int_{\pi + \pi/2}^{2\pi} \right] \cdots 
	 = \left[ \int_0^{\pi/2} + \int_{\pi/2}^{\pi + \pi/2} + \int_{-\pi/2}^{0} \right] \cdots \notag\\
	 & \qquad = \left[ \int_{-\pi/2}^{\pi/2} + \int_{\pi - \pi/2}^{\pi + \pi/2} \right] \cdots 
	 =\left[ \int_{-\pi/2}^{\pi/2} - \int_{- \pi/2}^{\pi/2} \right] \cdots = 0 \Longrightarrow I_1 = 0.
\end{align}
Note that from \eqref{ExSol2} we have with $n=1$ and $a_{1}\equiv K$, 
\begin{align}
\label{eq:integratedP}
R(t)=\frac{K}{2}\cos\left(\frac{2\pi t}{\Delta t}\right),
\end{align} 
which amplifies the first order of the Magnus expansion by a factor $\lambda=I_{0}(K)$.

\section{Error Analysis}\label{app4}

\begin{figure}
\centering
\includegraphics[width=0.50\columnwidth]{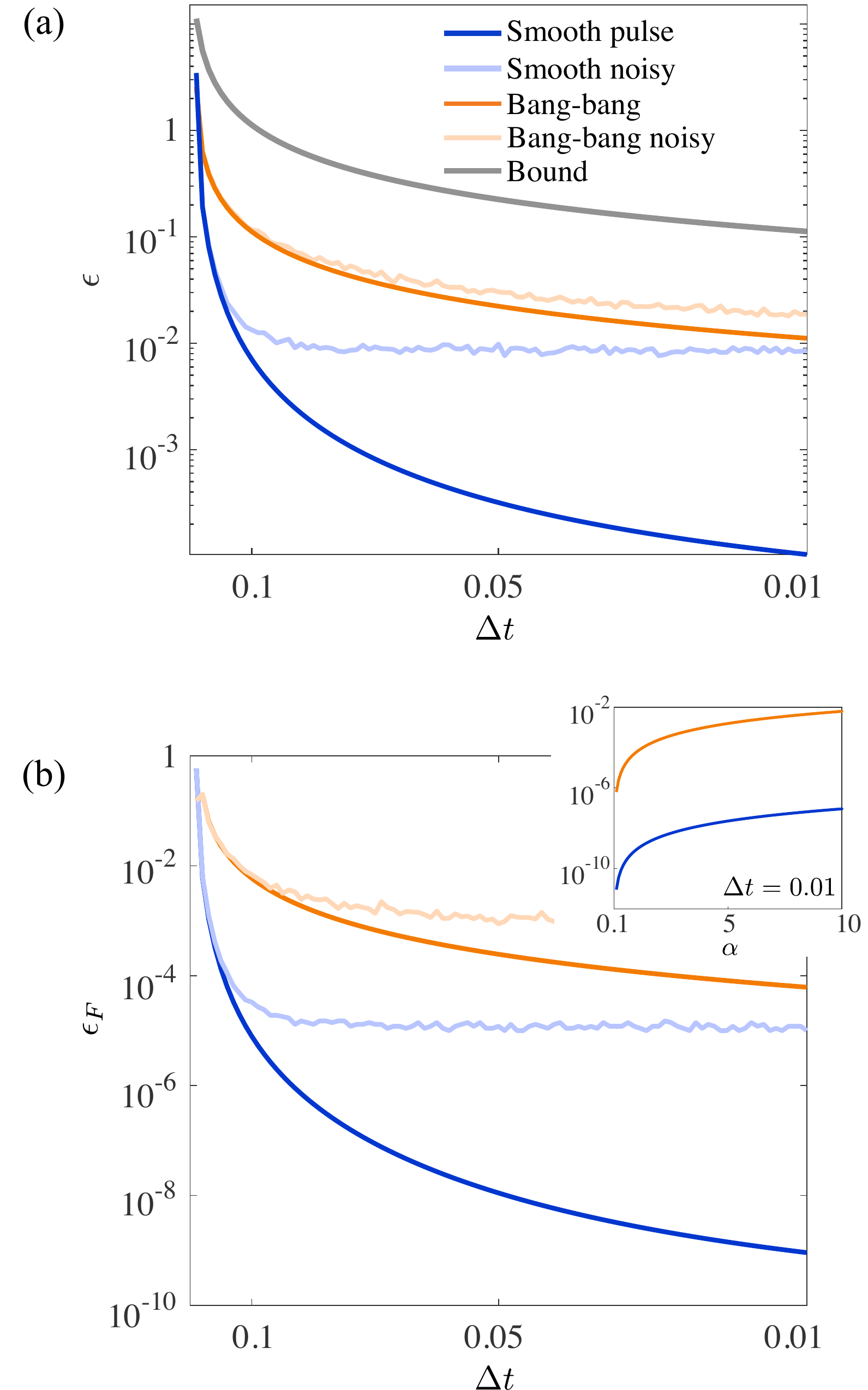}
\caption{\label{fig:error} Numerical study of Hamiltonian Amplification of a single harmonic oscillator with frequency $\omega=1/2$ for a fixed time $t=1$ and amplification by a factor $\lambda=2$. (a) error $\epsilon$ in the symplectic picture to obtain the amplified evolution $\exp(-\lambda A_{0}\Omega t)$ as a function of the spacing $\Delta t$ of the squeezing operation and (b) fidelity error $\epsilon_{F}=1-F$ with $F$ being the fidelity for Gaussian states explicitly given in \cite{Banchi} as a function of $\Delta t$ to prepare the coherent state $|\alpha\exp(-i\omega\lambda t)\rangle$ starting from an initial coherent state $|\alpha\rangle$ where $\alpha=1$. The inset plot in (b) shows for fixed $\Delta t$ the dependence of the fidelity error on $\alpha$.
In both case (a) and (b) the errors are evaluated for the case when the squeezing operations are applied 
instantaneously 
(orange curve) and in a smooth way (blue curve) described by the integrated pulse given in 
\eqref{eq:integratedP}. The soft orange curve shows amplification through 
bang-bang operations for which the squeezing phase is subjected to Gaussian 
noise with variance $\sigma=0.1$. The soft 
blue curve shows amplification through \eqref{eq:integratedP} containing 
Gaussian amplitude noise with mean zero and variance $\sigma=0.05$. Both noisy 
curves were averaged over $100$ trajectories. The solid grey curve represents 
the gate error upper bound for bang-bang operations derived in Sec \ref{app1} and explicitly given in the main body of manuscript. The figure shows a significant 
improvement in convergence in the case of smooth pulses, due to the cancelation 
of the second-order term in the Magnus expansion.}
\end{figure}

Any realistic pulse achieving Hamiltonian Amplification   
contains some noise, influencing the efficiency of the amplification process. To 
study the effect of noise we begin by considering the bang-bang sequence given 
by $S^{(\pm)}$. A generic squeezing transformation around an angle $\theta$ is described by  
\begin{align}
\label{eq:genericS}
S(\theta)=\exp\left(i\frac{r}{2}\left[\sin(\theta)(p^{2}-x^{2}
)+\cos(\theta)(xp+px)\right]\right),
\end{align}
noting that for $\theta=0$ and $\theta=\pi$ we obtain $S^{(\pm)}$, respectively. We model a first 
possible source of errors as fluctuations in the squeezing angle $\theta$ up to 
some error $\delta$. Then the squeezing in $x$ and $p$ 
direction perturbed by $\delta$ is described by 
$S_{\delta}^{(\pm)}=\exp(\pm 
i\frac{r}{2}[(xp+px)\cos(\delta)+(p^{2}-x^{2})\sin(\delta)])$. If we 
denote by $M_{\delta}$ the corresponding map formed by $\mathcal 
S_{\delta}^{(\pm)}$, the quadratic Hamiltonian $H_{0}$ given by Eq. (2) in the main paper 
transforms according to $M_{\delta}(H_{0})=\cosh(2r)H_{0}+\delta 
\sinh^{2}(r)H_{\text{er}}+\mathcal O(\delta^{2})$, where the first order 
error term is given by 
$H_{\text{er}}=\sum_{i,j}(\omega_{i,j}^{(x)}-\omega_{i,j}^{(p)})(x_{i}p_{j}+p_{i 
}x_{j})$. Thus, for small perturbations in the squeezing angle, the amplified 
Hamiltonian is linearly perturbed by $H_{\text{er}}$. However, we see that for 
quadratic Hamiltonians of the form 
$H_{0}=\sum_{i,j}\omega_{i,j}(x_{i}x_{j}+p_{i}p_{j})$ the error term vanishes 
to first order in $\delta$. A numerical example of the result is given by 
the soft orange curve in Fig.\,\ref{fig:error}.

As another example, we consider smooth 
amplification pulses that contain some noise. A broad class of noise can be 
described by including a (possibly unknown) Hamiltonian $H_{\text{er}}(t)$ 
describing the effect of the noise as possible random fluctuations in the total Hamiltonian. For a periodic 
controller the 1st order term of the Magnus expansion is then given by 
$\bar{H}^{(0)}=\frac{1}{2\Delta t}(\lambda H_{0}+H_{\text{noise}})$ where 
$H_{\text{noise}}=\int_{0}^{2\Delta 
t}dt^{\prime}U_{c}^{\dagger}(t^{\prime})H_{\text{er}}(t^{\prime})U_{c}(t^{\prime 
})$. 
Thus, when control is applied sufficiently fast, noise in the amplification 
pulses again adds linearly to the amplified Hamiltonian. In particular, the 
soft blue curve in Fig.\,\ref{fig:error} shows the error when the amplitude of the integrated pulse given in \eqref{eq:integratedP} suffers from Gaussian fluctuations about the desired 
time-dependent mean value (further details can be found in the caption). While for ideal 
bang-bang or smooth operations the amplification error
tends to zero when $\Delta t$ is reduced, we see that the noisy pulse and noisy 
bang-bang operations lead to a saturation of the amplification error at around 
$\epsilon\approx 10^{-2}$ (fidelity error $\epsilon_{F}\approx 10^{-5}$). The value at which the amplification error saturates 
is determined by the corresponding error term, which does not vanish when 
$\Delta t$ is reduced. 

As shown in Sec. \ref{app2}, while the error $\epsilon$ given by the Hilbert Schmidt distance of the amplified evolution in the symplectic representation $\exp(-\lambda A_{0}\Omega t)$ and the evolution obtained for a finite spacing $\Delta t$ of the squeezing operations can be upper bounded (solid grey curve in Fig. 1 (a)) using well known bounds for the Trotter sequence; a similar bound is challenging to obtain when the fidelity error $\epsilon_{F}$ for speeding up the preparation of Gaussian states is considered. However, it is interesting to note that both errors shown in Fig. 1 (a)
 and (b), respectively, indicate the same asymptotic behavior when $\Delta t$ is reduced. Furthermore, the inset plot in Fig. 1 (b) shows the dependence of the fidelity error on the expectation of $H_{0}$ with respect to the initial state, here taken to be a coherent state $|\alpha\rangle$. The development of an upper bound for $\epsilon_{F}$ in terms of $\Delta t$ capturing this behavior is left for future studies.

\end{document}